\documentclass[12pt, reqno]{amsart}
\usepackage[margin=3.6cm]{geometry}

\usepackage[utf8]{inputenc}
\usepackage[english]{babel}
\usepackage{amsmath, amsthm,amssymb,color}
\usepackage{amsfonts,dsfont,enumitem}

\newtheorem{Theorem}{Theorem}
\newtheorem{prop}{Proposition}[section]

\newtheorem{Lemma}[prop]{Lemma}

\theoremstyle{definition}
\newtheorem{ex}{EXAMPLE}[section]

\newtheorem{rem}{Remark}[section]
\newtheorem{defn}[prop]{Definition} 
\newtheorem{assu}{Assumptions\hfill \\}[section]

\newcommand{\C}{\mathbb{C}}

\newcommand{\Mbb}{\mathbb{M}}
\newcommand{\Abb}{\mathbb{A}}

\newcommand{\Vcal}{\mathcal{V}}

\newcommand{\Acal}{\mathcal{A}}
\newcommand{\Jcal}{\mathcal{J}}

\newcommand{\deq}{\overset{\mathrm{def}}{=}}

\newcommand{\Fcal}{\mathcal{F}}

\newcommand{\Lcal}{\mathcal{L}}

\newcommand{\Hbb}{\mathbb{H}}

\newcommand{\N}{\mathbb{N}}
\newcommand{\R}{\mathbb{R}}
\newcommand{\eps}{\varepsilon}

\newcommand{\alp}{\alpha}

\newcommand{\un}{\mathds{1}}

\renewcommand{\leq}{\leqslant}
\renewcommand{\geq}{\geqslant}
\newcommand{\defeq}{\stackrel{\rm{def}}{=}}

\newcommand{\Kcal}{\mathcal{K}}
\newcommand{\id}{\mathrm{id}}
\newcommand{\tl}{\triangleleft}
\newcommand{\tr}{\triangleright}

\newcommand{\api}{a^{+,\mathrm{int}}}
\newcommand{\ape}{a^{+,\mathrm{ext}}}
\newcommand{\ami}{a^{-,\mathrm{int}}}
\newcommand{\ame}{a^{-,\mathrm{ext}}}
\newcommand{\psie}[1]{\psi_h^{{#1},\mathrm{ext}}}
\newcommand{\psii}[1]{\psi_h^{{#1},\mathrm{int}}}
\newcommand{\fpep}{\mathfrak{p}^{+,\mathrm{ext}}}

\newcommand{\fpem}{\mathfrak{p}^{-,\mathrm{ext}}}

\newcommand{\fpip}{\mathfrak{p}^{+,\mathrm{int}}}
\newcommand{\fqip}{\mathfrak{q}^{+,\mathrm{int}}}

\newcommand{\fpim}{\mathfrak{p}^{-,\mathrm{int}}}
\newcommand{\fqim}{\mathfrak{q}^{-,\mathrm{int}}}

  \title{Cauchy Data for 1D singular Schr\"odinger operators }
\author[L. Hillairet]
{Luc Hillairet}
\email{luc.hillairet@math.univ-orleans.fr}
\address{Institut Denis Poisson, \\
Orleans, France}
  
  \author[J.L. Marzuola]
{Jeremy L. Marzuola}
\email{marzuola@math.unc.edu}

\address{Mathematics Department, University of North Carolina \\
CB\#3255, Phillips Hall, Chapel Hill, NC USA}
 
\begin{document}
  
  \begin{abstract}
    We study semiclassical 1-D Schrödinger operators
    of the form $Pu = -h^2 u'' \,+\,x^\gamma W(x) u$ on a finite interval $[0,b]$ for $0 < \gamma \in \mathbb{R} \setminus \mathbb{Q}$. We show that that the
    WKB expansions of solution can be extended on $[h^{1-\eps},b]$, for any $\eps>0$. Using a different approximation near $0$
    and a matching procedure, we obtain the Cauchy Data at $0$ of such WKB solutions. This allows us to derive singular
    Bohr-Sommerfeld rules. We also pay special attention to uniformity in $W$ for our expansions.   
\end{abstract}
\maketitle

\section{Introduction}
We consider a self-adjoint realization of the one dimensional semiclassical Schr\"odinger operator 
\[
P_h u \,=\, -h^2 u''\,+\, V(x) u
\] 
that is defined on a interval $I\,=\,[0,b]$ with some boundary condition at $0$, $b$ and the potential $V$ is
defined by $x\mapsto x_+^\gamma W(x)$ for some $0 < \gamma \in \mathbb{R} \setminus \mathbb{Q}$ and smooth $W$. 
The eigenvalue equation 
\begin{equation}\label{eq:eigeq}
P_h u_h=E_h u_h, 
\end{equation} 
can be studied by asking that the Cauchy data at $0$ and $b$ of a solution $u_h$ (i.e. $(u_h(0),hu_h'(0))$ and $(u_h(b),hu_h'(b))$)
satisfy the boundary conditions. This approach requires to relate the two Cauchy data at both ends of the interval.
In smooth settings, WKB expansions can be used to make this relation explicit leading to the so-called Bohr-Sommerfeld rules, see \cite{Bender_Orszag78,de2005bohr} e.g..
In settings for which $0$ exhibits some kind of singularity, it is sometimes useful to split the interval
$[0,b]$ and to use expansions in $[0,b_h]$ and $[a_h,b]$ respectively.
Typically, the methods used to obtain the expansions
on $[0,b_h]$ and $[a_h,b]$ will be different and valid only in some regime (i.e. for some choice of $a_h,b_h$).
In order to obtain a full answer, it is then crucial that the latter intervals overlap for some choice of $a_h$ and $b_h$.
When the singularity is at $0$, the interval $[0,b_h]$ is usually called \textit{the interior region} or
\textit{the boundary layer}, the interval $[a_h,b]$ is \textit{the exterior region} and the interval
$[a_h,b_h]$ is \textit{the matching region}.

Obtaining the matching region usually requires to go beyond standard analysis. More precisely,
in our case, we will use WKB expansions in $[a_h,b]$. For the latter, the regime for which $a_h$ of order $h^0$ is
standard. The main task is thus to push the method further so as to obtain expansions valid on $[h^\alpha,b]$ for some
positive $\alpha$. On $[0,b_h]$ we will use the classical variation of constants method.
In both cases, we will obtain joint asymptotic expansions for $(u(x_h), hu(x_h))$ when $h$ and $x_h$ go to zero. 

Our main result is as follows.  Let $\phi_h$ be a solution to our Schrödinger equation. Since the space of solutions is of dimension $2$, there is a
linear relation between the Cauchy datum at $0$ and the Cauchy datum at $b$. The idea of matching is
to use an intermediate interval $[h^{1-\eps_0},h^{1-\eps_1}]$ where $0 < \eps_1 < \eps_0 < 1$. For $x$ in this interval, we use
the interior solutions that we will construct to relate the Cauchy data at $0$ and $x_h$ and
 WKB solutions to relate the Cauchy data at $x_h$ and $b$. Basic linear algebra and careful asymptotic
analysis will then yield the following theorem.

\begin{Theorem}\label{thm:matching}
  Take $0<\gamma \notin \mathbb{Q}$.  There exist matrices $\mathbb{A}^{\pm}_h(E)$ that admit
  an asymptotic expansion with exponent set $\{m\gamma +n ,~m\geq 0, n\geq 0\}\setminus \{0\}$ such that
  for any solution $\phi$  to the Schrödinger equation \eqref{eq:eigeq}, the following relation holds:
  \begin{multline*}
    \begin{pmatrix}
      E^{\frac{1}{4}}\phi(0)\\
      E^{-\frac{1}{4}}h\phi'(0)
    \end{pmatrix}
    \,=\,\\
    \left (D_h\,+\,\cos (\frac{\sigma_E}{h})\cdot \Abb^+_h(E)\,+\,\sin (\frac{\sigma_E}{h})\cdot \Abb^-_h(E)\right)
    \begin{pmatrix}
      (E-V(b))^{\frac{1}{4}}\phi(b)\\
      (E-V(b))^{-\frac{1}{4}}h\phi'(b)
    \end{pmatrix}
  \end{multline*}
  where
  \[
    D_h\,=
    \begin{pmatrix}
      \cos \frac{\sigma_E}{h} & -\sin \frac{\sigma_E}{h} \\
      \sin\frac{\sigma_E}{h} & \cos \frac{\sigma_E}{h}
    \end{pmatrix}
    ~~\text{and}~~
    \sigma_E\,=\,\int_0^b \sqrt{E-V(y)}\, dy.
 \]
\end{Theorem}

\begin{rem}
    The case $\gamma \notin \mathbb{Q}$ is a technical requirement to keep the different exponent sets that appear from including the value $-1$. 
    If $\gamma \in \mathbb{N}$, the result follows from standard WKB analysis. If $\gamma \in \mathbb{Q}\setminus \mathbb{N}$, similar results should 
    hold including $x^{m\gamma +n}\log x$ terms in the representative expansion.    
\end{rem}

The study of Schr\"odinger operators is a standard, very classical problem and many properties of their eigenvalues
and eigenfunctions can be found in the literature on Sturm-Liouville problems and semiclassical analysis
(see Titchmarsch \cite{titchmarsh1946eigenfunction}, Olver \cite{olver}, H\"ormander \cite{Hor-v1,Hor-v2,Hormander3,Hormander4},
Maslov \cite{Maslov72}, 
Helffer-Robert \cite{helffer1983calcul}, Dimassi-Sj\"ostrand \cite{DiSj}, Zworski \cite{zworski2012semiclassical}). 

Related rules
for smooth potentials ($\gamma \geq 2$) in the semiclassical literature  for a sequence of
eigenvalues $(E_h)_{h>0}$ that converges to a {non-critical energy $E_0 > 0$} with a connected energy surface can be found in Section 10.5 in \cite{Bender_Orszag78} or \cite{HMR87, de2005bohr, Yafaev11}.

This is a follow-up result and related to the authors' previous results on eigenvalue spacings for Schrödinger operators with rough potentials \cite{hillairet2023eigenvalue}. There we considered $b = + \infty$ and analyzed the eigenvalue spacings resulting from boundary conditions at $x=0$ using very different techniques, such as the construction of semiclassical defect measures.  Notably, the spacings found in \cite{hillairet2023eigenvalue} depend upon the singular potential parameter $\gamma$ in a natural way, related to how the exponent set determines the behavior of the matrix $\Abb_h^{\pm}$. The Eigenvalue spacings for different boundary conditions can also be inferred from the Bohr-Sommerfeld rules through the matrix equation we establish in Theorem \ref{thm:matching}.



As seen in \cite{hillairet2012nonconcentration}, the potentials we consider here arise from the adiabatic ansatz in a stadium-like
billiard. In addition, semiclassical Schr\"odinger operators of this sort appear in the study of waveguides with corners \cite{rouvinez1995scattering, DaugeRaymond2012waveguides}, of flat triangles \cite{OurmieresBonafos2015triangles, HillairetJudge2011triangles}, and of diffractive trapping for conormal potentials \cite{GW}.  Singular potentials have also been studied in for instance \cite{LaiRobert79, Berry82, Chr_AIF_15,filippas2023whispering}.  See also \cite{friedlander2009spectrum} and \cite{Simon_lowlying_83} for a much more complete study of the bottom of the well for quadratic potentials ($\gamma = 2$), or \cite{Bony_Popoff19} for even more degenerate situations. The study of semi-excited states in Sj\"ostrand \cite{Sjostrand_semi92} is also related.

The paper will proceed as follows.  In Section \ref{sec:setting}, we clearly define the problem and lay out the notation necessary to proceed.  Then, in Section \ref{sec:WKBext}, we describe a WKB expansion that is valid up to an $h$ dependent neighborhood of $0$, i.e. the {\it exterior region}.  Then, we control the eigenfunctions in a small $h$-dependent neighborhood of $0$ in Section \ref{sec:int}, i.e. the {\it interior region}.  In Section \ref{sec:matching}, we prove Theorem \ref{thm:matching} by gluing solutions together on interface of the exterior and interior regions.  Lastly, in Section \ref{sec:examples}, we apply Theorem \ref{thm:matching} to two key settings of computing singular Bohr-Sommerfeld rules.

\subsection*{Acknowledgments} This work initiated when the second author
visited the first for an extended stay as a professeur invit\'e at the Universit\'e d'Orl\'eans and also benefited from the invitation of the first author
to the UNC at Chapel Hill. The authors thank both institutions. J.L.M. acknowledges supports from the NSF through NSF grant DMS-2307384, and L.H. acknowledges the support of the ANR program ADYCT (grant ANR-20-CE40-0017). We warmly thank the anonymous referee who carefully read the first version of the paper.

\section{Setting}
\label{sec:setting}
We consider the semiclassical Schrödinger equation
\begin{equation}\label{eq:Sch}
  -h^2u'' \,+\,V(x)u \,=\,E_h u.
\end{equation}
on the interval $[0,b]$.

Let $(u_h)_{h\leq 1}$ be a family of solutions to \eqref{eq:Sch}, and $(x_h)_{h\leq 1}$ be such that
$x_h$ decreases to $0$ for $h$ going to $0$.
       We define the semiclassical Cauchy datum at $x_h$ by
       \[
         C_h\,=\,
         \begin{pmatrix}
           u_h(x_h) \\
           hu'_h(x_h)\\
         \end{pmatrix}.
       \]

For a compact subinterval $I \subset [0,b]$, we equip the space $C^\infty (I)$ with its classical
Fréchet topology associated with the family of norms $(p_N)_{N\geq I}$ defined by 
\[
  \forall u\in C^{\infty}(I),~~p_N(u)\,=\,\max \left \{ \sup \{ |u^{(k)}(x)|,~x\in I \},~0\leq k \leq N \right \}.
\]
Observe that the notation does not reflect the dependence on $I$ that will be clear from the context.

\begin{assu}
\label{assu:bot}
We make the following assumptions on $V$:
\begin{itemize}
\item The potential $V$ is smooth and increasing on $(0,b]$ and continuous on $[0,b]$.
\item $V(0)=0$ and there exist $\gamma >0$ and $W>0$ smooth on $[0,b]$ such
  that $\forall x>0,~ V(x) = x^\gamma W(x).$
\end{itemize}
\end{assu}

We fix $\Kcal$ a compact set in $C^\infty([0,b] ; \R)$, and denote by $\Vcal$ the set of potentials $V$ that
satisfy the preceding assumptions with $W\in \Kcal$. We fix $K$ a compact interval in $(0,+\infty)$ and
we assume that

\begin{equation}\label{eq:defdelta}
  \left \{
    \begin{array}{c}
      \exists \delta >0,~~\forall (V,E) \in \Vcal\times K , \\
      \forall x\in [0,b],~~E-V(x)\,\geq \delta.\\
    \end{array}
  \right.
\end{equation}

This assumption says that for any energy $E \in K$, the interval $[0,b]$ is in the classically allowed region.
We can thus perform a WKB approximation for $u_h$. It is expected that, when the potential is not smooth at $x=0$,
the WKB method will give a good approximation of the solution $u_h$ only for $x \geq a_h>0$. In the following section, we provide
the necessary estimates to give quantitative statements about $a_h$ and the corresponding asymptotic expansions for $C_h$.

\section{WKB Ansatz in the exterior region}
\label{sec:WKBext}

The WKB method (see for instance Dyatlov-Zworski \cite{dyatlov2019mathematical}, Zworski \cite{zworski2012semiclassical} and many others) gives
asymptotic expansions for any solution to the second order ODE
\[
  h^2 u_h'' \,+\, q(x) u_h \,=\,0, 
\]
on some interval $I_h\deq [a_h,b_h]$, where $q$ is a smooth potential that is positive on $I_h$.
In our setting, we have $q=E-V$.

The strategy consists in first constructing two independent ($O(h^\infty)$) quasimodes $u_h^{\pm}$ and then proving that
any true solution is $O(h^\infty)$ close to a linear combination of the $u_h^{\pm}$.
It is usually performed with a smooth potential $q$ on a fixed interval $I$ (i.e. $I_h$ independent on $h$).
In that case, both $O(h^\infty)$ remainder terms can be estimated using the sup-norm over $I$ of $q^{-1}$ and of the derivatives of
$q$. A rather crude estimate (or even knowing that such an estimate exists) is enough to ensure that the method
works.

With our approach, it will be crucial to let $a_h$ go to zero so that, when $\gamma$ is not an integer, the sup-norm
on $I_h$ of high-order derivatives of $q$ will blow-up. As a result, the sequence $a_h$ cannot decrease too fast to zero but
it is crucial to our method that $a_h$ does not decrease too slow either. Indeed, the main task here is to determine
the greatest $\alpha$ such that the WKB expansion holds on $[h^\alpha,b]$.

For the convenience of the reader, we have found it clearer to present the basics of the WKB method so as to see
what estimate is needed.  The WKB Ansatz consists in writing $u_h$ under the following form :
\[
  u_h (x)\sim \exp(\frac{i}{h}S(x))\sum_{k\geq 0} h^k A_k(x).
\]
Plugging into the equation and putting together the terms with the same power of $h$, we obtain the following
set of equations.
\begin{itemize}
\item The eikonal equation:
  \[
\forall x\in I_h,~S'(x)^2 \,=\, q(x).
  \]
\item The homogeneous transport equation:
  \[
    \forall x\in I_h,~ 2S'(x)A_0'(x)\,+\,S''(x)A_0(x)\,=\,0.
  \]
\item The inhomogeneous transport equations:
  \[
    \forall k\geq 0,~2S'A_{k+1}'\,+\,S''A_{k+1}\,=\,-iA_k''.
  \]
\end{itemize}

If this system can be solved, then for any solution $(S,(A_k)_{k\geq 0})$ and any $N$, we can define
$u_{h,N}^+\,=\,\exp(\frac{i}{h}S)\sum_{k=0}^N h^k A_k$ and this function then satisfies
\[
  h^2{u_{h,N}^+}'' \,+\, q \cdot u_{h,N}^+\,=\, h^{N+2}A_{N}''\exp(\frac{i}{h}S).
\]

The eikonal equation can be solved because $q$ is positive. The homogeneous and inhomogenous transport equations are linear first order ODE
that thus can also be solved. We choose the following solution
\begin{gather}
\nonumber  \forall x \in I_h, \ \ S'(x) \,=\, \sqrt{q(x)}, \ \ S(x)\,=\,-\int_x^b S'(y)\,dy,\\
\nonumber  A_0(x)\,=\, [S'(x)]^{-\frac{1}{2}}\,=\, [q(x)]^{-\frac{1}{4}},\\
\label{def:WKB} \forall k\geq 0,~A_{k+1}(x)\,=\,\frac{i}{2}A_0(x)\cdot \int_x^{b}A_k''(y)(A_0(y))^{-1}\, dy.
\end{gather}  
With this choice, we define $u^+_{h,N}$ as above and set $u_{h,N}^- \deq\, \overline{u_{h,N}^+}.$

We now proceed to make estimates in the case that $q(x)=E-V(x)$ and the interval $I_h \subset [a_h,b]$. 
When the potential $V$ is smooth on $[0,b]$ the functions $A_k$ determined by the preceding formulas are also smooth
on $[0,b]$. It is no longer the case when $\gamma \notin \N$ for which we need a convenient setting of
asymptotic expansions at $0$.

\subsection{Generalized Taylor expansions}
The idea behind dealing with singular potentials is to keep track separately of the singular behavior
when $x$ goes to zero. When differentiating again and again, new singular powers of $x$ appear.
In order to be able to follow each of them, it is convenient to set the following definition.

\begin{defn}
  Let $u$ be a continuous function on some interval $(0,b]$. We will say that $u$ admits a generalized Taylor expansion at $0$ if
  there exists a discrete set $\Acal \subset \R$ that is bounded from below, and a collection of complex numbers $(a_\alpha)_{\alpha \in \Acal}$
  such that
  \[
    \forall N\geq 0,~\exists C_N,~ \forall x\in (0,b],~\left |u(x)-\sum_{\alpha\in \Acal, \alpha <N}a_\alpha x^\alpha \right |\,\leq \, C_N x^N.
  \]
  We will say that $\Acal(u)$ (or simply $\Acal$ if there is no ambiguity) is the exponent set of $u$.
\end{defn}

\begin{rem}
    If we wish to deal with $\gamma \in \mathbb{Q}$, then we would need to use generalized Taylor expansion with respect to a {\it scale} of functions that also include the functions $x\mapsto x^\alpha \log x$.
\end{rem}

Observe that there is a small ambiguity in the set $\Acal$. Indeed, we can artificially add exponents and say that the
corresponding coefficient vanishes.
However, the set of $\alpha$ for which $a_\alpha \neq 0$ is determined by $u$. Indeed, either $u = O(x^\infty)$ or we have
\[
  \min \{ \alpha\in \Acal, a_\alpha\neq 0\}\,=\, \sup \{ \alpha \in \R,~ \lim_{x\rightarrow 0} x^{-\alpha}u(x) \,=\,0\}.
\]
Once $\alpha_0$ is determined, we get $a_0\,=\, \lim\limits_{x\rightarrow 0} x^{-\alpha_0}u(x)$.
Inductively, we then obtain a sequence $\alpha_n$ and the corresponding $a_n$.
This argument shows that if we know that $u$ has some generalized Taylor expansion then we can find some exponent set that is associated
to it. Alternatively, if we know \textit{a priori} the exponent set then, as for regular Taylor expansions,
each coefficient is determined by $u$ (including those that vanish). 

We record here a few facts that generalize the corresponding statements for usual Taylor expansions.
The proofs are left to the reader.
\begin{itemize}
\item If $u$ and $v$ admit generalized Taylor expansions then so does any linear combination of $u$ and $v$ and  
\[
\forall \lambda, \mu \in \C,~ \Acal(\lambda u+ \mu v)\,=\,\Acal(u)\cup \Acal(v).
\]
\item If $u$ and $v$ admit generalized Taylor expansions, then so does the product $uv$ and
  \[
\Acal(uv)\,=\big\{ \alpha+\beta,~\alpha \in \Acal(u),\, \beta \in \Acal(v) \big \}.
\]
\item If $u$ admits a generalized Taylor expansion, and $-1\notin \Acal(u)$, then the function $U$ defined on $(0,b]$
  by $U(x)\,=\, \int_x^b u(y) \, dy$ admits a generalized Taylor expansion and
  \[
    \Acal(U)\,=\big \{ \alpha +1,~ \alpha \in \Acal(u)\big \} \cup \{0\}.
  \]
  This property implies that if $u$ and $u'$ admit generalized Taylor expansion then
  \[
    \Acal(u')\,=\, \big \{ \alpha -1,~ \alpha \in \Acal(u)\setminus \{0\} \big \}.
  \]
  In particular, $-1$ is never in $\Acal(u')$.
\item If $u$ admits a generalized Taylor expansion at $0$ and $\alpha_0 = \min \big \{ \alpha \in \Acal(u)\setminus \{0\}, a_\alpha \neq 0 \big \}$.
  If $\alpha_0 \geq 0$, $0 \in \Acal(u)$, and $a_0\neq 0$ then there exists a function $v$ that is continuous on $[0,b]$ such that
  \[
    \forall x\in (0, b],~ u(x)\,=\, (1+x^{\alpha_0})v(x).
  \]
  In any other cases, there exists a function $v$ that is continuous on $[0,b]$ such that
   \[
    \forall x\in (0, b],~ u(x)\,=\, x^{\alpha_0}v(x).
  \]
\end{itemize}

We observe that such a definition is quite common in asymptotic analysis.
It is also closely linked with asymptotic expansions in symbol classes for conormal potentials
(see \cite{GW} Gannot-Wunsch).

\begin{ex}\label{ex:gte}
  For any $(W,E)\in \Kcal\times K$ and any $\alpha \in \R$, the function $q_\alpha$ defined on $(0,b]$
  by $q_\alpha(x)\,=\,(E-x^\gamma W(x))^\alpha$ admits a generalized asymptotic expansion with exponent set
  $\{ m\gamma +n, m\geq 1, n\geq 0 \} \cup \{0\}.$
  Indeed, using the fact that $(E-x^\gamma W)$ never vanishes, we can make a Taylor expansion
  \[
    q_\alpha(x)\,\sim\, E\,+\,\sum_{m\geq 1} c_{\alpha,m} x^{m\gamma} (W(x))^m. 
  \]
  The claim follows by making a Taylor expansion of $W^m$.
\end{ex}

The rest of this section is devoted to prove the following proposition.

\begin{prop}\label{prop:estAk}
   In the setting described above, let $(A_k)_{k\geq 0}$ be the sequence of functions defined on $(0,b]$ by \eqref{def:WKB}.
  \begin{enumerate}
  \item For any $k,\ell\ \geq 0$, $A_k^{(\ell)}$ admits a generalized Taylor expansion at $0$.
    Defining $\Acal_{k,\ell}$ the exponent set of $A_k^{(\ell)}$, we have
      \begin{gather*}
       \Acal_{k,0} \deq \left \{ m\gamma + n-k,~ m\geq 1,~ n\geq 0 \right \}\cup \{0\},\\
      \forall \ell \geq 1,~ \Acal_{k,\ell} \deq \left \{ m\gamma + n-k-\ell,~ m\geq 1,~ n\geq 0 \right \}.
    \end{gather*}
  \item For any $k\,\geq 0$, there exists $C_k\in\R,\,N_k,\,\sigma_k \in\N$ such that, for all $W\in \Kcal$ and $E\in K$,
    \[
    \forall x\in (0,b],~ |A_k(x)|\,\leq C_k(1+x^{\gamma-k})(1+p_{N_k}(W))^{\sigma_k} . 
    \]
  \item For any $k\geq 0$, for any $\ell \geq 1$, there exists $C_{k,\ell}\in\R,\,N_{k,\ell},\,\sigma_{k,\ell} \in\N$ such that, for all
    $W\in \Kcal$ and $E\in K$,
    \[
    \forall x\in (0,b],~ |A_k^{(\ell)}(x)|\,\leq C_{k,\ell}\cdot x^{\gamma-k-\ell}\cdot (1+p_{N_{k,\ell}}(W))^{\sigma_{k,\ell}} . 
    \]
    
\end{enumerate}
\end{prop}

The proof of this proposition will be by induction and we begin by studying the case $k=0$.

\subsection{A preliminary estimate}
The following lemma will allow us to control the derivatives of $A_0$.

\begin{Lemma}\label{lem:init}
  Let $W_0$ be a smooth function defined in a neighborhood of $0$ and $\gamma \in \R \setminus \mathbb{Q}$. Let $E$ and $b$ be such that the function
  $x\mapsto E-x_+^\gamma W_0(x)$ is positive on $[0,b]$. For $\alpha\in \R$, define $q_\alpha$ on $[0,b]$ by
  $q_\alpha(x)\,=\, \left[ E-x_+^\gamma W_0(x) \right ]^{\alpha}$.
  Then there exist functions $W_{\alpha,\ell, j,m,n}$ that are smooth on $[0,b]$ such that, for all $\ell \geq 1$ and all $x\in (0,b]$ we can write
  \[
    q_{\alpha}^{(\ell)}(x)\,=\,\sum_{j=1}^\ell \left[E-x_+^\gamma W_0(x)\right ]^{\alpha-j}\sum_{m=1}^{\ell}\sum_{n=0}^{\ell} x^{m\gamma -n}W_{\alpha,\ell,j,m,n}(x).
  \]
If we define  
  \[
p_N(W_{\alpha,\ell,\bullet})\,=\, \max \{ p_N(W_{\alpha, \ell, j,m,n}),~ 1\leq j,m\leq \ell,~0\leq n\leq \ell \},
\]
then for any $\ell$ and any $N$ there exists a constant $C\deq C(\ell,N)$ such that 
  \[
p_N(W_{\alpha,\ell,\bullet}) \leq C (1+p_{N+1}(W_0))(1+p_{N+2}(W_0))\cdots (1+p_{N+\ell}(W_0))p_{N+\ell}(W_0).
\] 
Moreover, if $\gamma$ is an integer, then $W_{\alpha,\ell,j,m,n}$ vanishes as soon as $m\gamma-n<0$. 
\end{Lemma}

In the sequel, we will need this lemma only for $\alpha = \pm \frac{1}{4}$.

\begin{proof}
  In order to make the notations a bit lighter, we omit the dependence with respect to $\alpha$ below. Thus we set
  $W_{\ell,j,m,n}=W_{\alpha,\ell,j,m,n}$.
  
  The proof is by induction on $\ell$. The fact that $q_\alpha^{(\ell)}$ has the given expression is obtained by a straightforward derivation.
  Indeed, we find that 
  \begin{multline*}
    W_{\ell+1,j,m,n}\,=\, W'_{\ell,j,m,n}\,+\,(m\gamma-n+1)W_{\ell,j,m,n-1}\\
    +\,(\alpha-j+1)\left( W_0'\cdot W_{\ell,j-1,m-1,n}\,+\,\gamma\cdot W_0\cdot W_{\ell,j-1,m-1,n-1}\right)
  \end{multline*}
  with the convention that if $j,m,n$ is not in the range given for the sum defining $q_\alpha^{(\ell)}$ then the corresponding $W_{\ell,j,m,n}=0$. 
  Using the Leibniz derivation rule, the preceding expression also gives some $C$ (that depends only on $\ell$ and $N$) such that
  \[
    \begin{split}
      p_{N}(W_{\ell+1,\bullet})\,& \leq \,C\big( p_{N+1}(W_{\ell,\bullet})\,+\,p_{N+1}(W_0)p_{N}(W_{\ell,\bullet})\big)\\
      &\leq \,C\big( 1+p_{N+1}(W_0)\big)\cdot p_{N+1}(W_{\ell,\bullet}).
\end{split}
  \]
  The estimate follows by induction.
\end{proof}

\subsection{Proof of Proposition \ref{prop:estAk}}
\subsubsection{The case $k=0$}
Using the notations of Lemma \ref{lem:init}, we have
\[
  A_0=q_{-\frac{1}{4}}.  
\]

The fact that $A_0$ admits a generalized Taylor expansion follows from Example \ref{ex:gte} and the bound follows from the
fact that
\[
  \forall x\in [0,b],~~|A_0(x)|\,\leq\, \delta^{-\frac{1}{4}},
\]
where $\delta$ has been defined in \eqref{eq:defdelta}.

We now estimate the derivatives $A_0^{(\ell)}$, starting from the expression in Lemma \ref{lem:init}.
Since, for any $j$, $z\mapsto (E-z)^{\alpha-j}$ has a power series expansion in a neighborhood of $0$,
we can expand $(E-x^\gamma W(x))^{\alpha-j}\,=\,\sum_{m\geq 0} a_{j,n}x^{m\gamma}(W(x))^m$. Expanding $x\mapsto W(x)^m$
in Taylor series gives a generalized Taylor expansion for $x\mapsto (E-x^\gamma W(x))^{\alpha-j}$
whose exponent set is $\{ m\gamma+n,~m\geq 1, n\geq 0\}\cup \{0\}$. The claim then follows using the properties of
functions with generalized Taylor expansions.

We now prove the estimate for $A_0^{(\ell)}, \ell\geq 1$. Using properties of generalized Taylor expansions, there exists a function $B_{0,\ell}$
that is continuous on $[0,b]$ and such that  
\[
A_0^{(\ell)}(x)\,=\,x^{\gamma-\ell}B_{0,\ell}(x).
\]
We set $b_{0,\ell} \,=\, \|B_{0,\ell}\|_{\infty}$.

The expression in Lemma \ref{lem:init} implies that
\begin{equation*}
  \forall x\in (0,b],\, |A_0^{(\ell)}(x)| \, \leq p_0(W_{\ell,\bullet})\cdot 
  \sum_{j=1}^\ell \delta^{-\frac{1}{4}-j}\sum_{m=1}^\ell\sum_{n=0}^\ell x^{m\gamma-n}. 
\end{equation*}
We observe that the smallest power that appears is $\gamma-\ell$ so that we can factorize it. The remaining sum is then
bounded by $Cp_0(W_{\ell,\bullet})$. This gives the result we want for $b_{0,\ell}$, given the estimate on $p_0(W_{\ell,\bullet})$ provided by Lemma \ref{lem:init}.

\subsubsection{The induction step}
Let us now prove that if for any $\ell,~A_k^{(\ell)}$ has a generalized Taylor expansion, then so does $A_{k+1}^{(\ell')}$, for any $\ell'$.

Using the notations of Lemma \ref{lem:init}, we have $A_0^{-1}=q_{\frac{1}{4}}$ and the latter admits a generalized Taylor expansion
(see Example \ref{ex:gte}). So, using the induction hypothesis, $y\mapsto A_k''(y)A_0^{-1}(y)$ has a generalized Taylor expansion.
Since $\gamma \notin \mathbb{Q}$, $-1$ is not in the exponent set of the latter function, it follows that $x\mapsto \int_x^b A_k''(y)A_0^{-1}(y) \,dy$
also has a generalized Taylor expansion Moreover, the exponent set is seen to be
  \[
    \{ m\gamma+n-k-1,~m\geq 1, n\geq 0\} \cup \{ 0\}.
  \]
  
We now use Leibniz derivation rule and observe that $A_{k+1}^{(\ell')}(x)$ can be written as a linear combination of terms of the following form :
  \begin{gather*}
    A_0^{(\ell')}(x) \cdot \int_x^b   A_k''(y)A_0^{-1}(y) \,dy, \\
    A_0^{(\ell_1)}(x)A_k^{(2+\ell_2)}(x)(A_0^{-1})^{(\ell_3)}(x),~~\ell_1+\ell_2+\ell_3+1=\ell'.
  \end{gather*}
  All these terms have a generalized Taylor expansion using the induction hypothesis, Lemma \ref{lem:init} 
  the statement for $A_0^{(\ell)}$ and $(A_0^{-1})^{(\ell)}$ and the preceding argument
  for the integral. The exponent set is easily derived. We thus obtain the first claim of the proposition.

  We now move to prove the remaining estimate. Using the properties of functions with generalized Taylor expansion, we define continuous functions
  $(B_{k,\ell})_{k,\ell \geq 0}$ such that
  \begin{gather*}
    \forall k\geq 0,~\forall x\in (0,b], A_k(x)\,=\,(1+x^{\gamma-k})B_{k,0}(x),  \\
    \forall k\geq 0,\, \forall \ell\geq 1,~\forall x\in (0,b],~A_k^{(\ell)}(x)\,=\,x^{\gamma-k-\ell}B_{k,l}(x).
\end{gather*}
We will also denote by $b_{k,\ell}\deq \|B_{k,\ell}\|_{\infty}$. Observe that the statement of the proposition is equivalent to
proving that there exist constant $C,N,\sigma$, independent of $W\in \Kcal$ and $E\in K$ such that
\[
  b_{k,\ell}\,\leq \, C(1+p_N(W))^\sigma.
\]
This is again proved by induction on $k$.

From the definition of $A_{k+1}$, we derive
\[
  A_{k+1}(x) \leq\,\delta^{-\frac{1}{4}}(\sup K +b^\gamma \sup_{W\in \Kcal}\|W\|_{\infty})^{\frac{1}{4}} \int_x^b y^{\gamma-k-2}b_{k,2}\, dy.
\]
Observe that since $\gamma-k$ is never $0$, for any $k$, $x\mapsto (1+x^{\gamma-k-1})^{-1}\int_{x}^b y^{\gamma-k-2}$ is bounded on $[0,b]$.
This gives the relation
\[
b_{k+1,0}\, \leq \, C\cdot b_{k,2} p_0(W)
\]
with a constant $C$ that depends only on $\delta,$ and $ \gamma-k$.
We now assume $\ell'\geq 1$ and address $b_{k+1, \ell'}$. We address all the terms that appear in the formula for
$A_{k+1}^{(\ell)'}$, namely
\begin{gather}
  A_0^{(\ell')}(x) \cdot \int_x^b   A_k''(y)A_0^{-1}(y) \,dy \,\leq\, C \cdot x^{\gamma-\ell'}(1+x^{\gamma-k-1})b_{0,\ell'}b_{k,2}p_0(W).
\end{gather}

For the terms of the form  $A_0^{(\ell_1)}(x)A_k^{(2+\ell_2)}(x)A_0^{(\ell_3)}(x),~~\ell_1+\ell_2+\ell_3+1=\ell'$, 
we study four cases, depending on whether $\ell_1,\,\ell_3$ vanish or not. We obtain the following bounds
(up to a uniform multiplicative constant)
\[
\left \{ \begin{array}{ll}
  x^{\gamma-k-\ell'-1}b_{k,\ell'+2}, & \ell_1=0,\,\ell_3=0, \\
  x^{2\gamma-k-\ell'-1}b_{k,\ell_2+2}b_{0, \ell_3}, & \ell_1=0,~\ell_3 \neq 0, \\
  x^{2\gamma-k-\ell'-1}b_{0,\ell_1}b_{k,\ell_2+2}, & \ell_1\neq 0,~ \ell_3=0, \\
  x^{3\gamma-k-\ell'-1}b_{0,\ell_1}b_{k,\ell_2+2}b_{0,\ell_3},  & \ell_1\neq0, \ell_3\neq 0 .
\end{array}
\right .
\]
Comparing all the terms, we see that, if we factorize $x^{\gamma-\ell'-k-1}$, the remaining
powers of $x$ will be non-negative. Finally, we obtain the crude estimate
\[
  b_{k+1,\ell'} \leq C \big( \max \{ b_{0,\ell_1},~\ell_1\leq \ell'\}\big)^2\max\{ b_{k,\ell_2},~\ell_2 \leq \ell'+1 \}.
\]
This estimate is good enough to obtain the claimed result by induction on $k$.

\subsection{Solutions in the exterior region}
In this section, we show that any true solution can be approximated by WKB constructions on intervals
$[a_h,b]$ for good choices of $a_h$. Let $\psi$ be an exact solution to
\[
  h^2 \psi''\,+\,\left(E-x^\gamma W(x)\right)\psi\,=0 .
\]
Elaborating on the variation of constants methods we look for functions $A^{\pm}$ such that
\begin{equation}\label{eq:varconst}
  \left \{
    \begin{array}{cl}
      \psi \,&=\, B_+ u_{h,N}^+\,+\,B_-u_{h,N}^- ,\\
      \psi' \,&=  B_+ (u_{h,N}^+)'\,+\,B_-(u_{h,N}^-)',\\
    \end{array}
    \right .
\end{equation}
where $u_{h,N}^{\pm}$ are given by the WKB construction. More precisely, we define

\[
  u_{h,N}^+(x)\,=\,\exp\left( \frac{iS}{h}\right )\left( A_0(x)\,+\,\sum_{k=1}^N h^kA_k(x)\right)
\]
and $u_{h,N}^-= \overline{u_{h,N}^+}$.

With this choice, we have
\[
  (u_{h,N}^\pm)''\,+\,(E-x^\gamma W(x))u_{h,n}^\pm\,=\,r_{h,N}^\pm 
\]
and 
\[
  |r_{N,h}^{\pm}|\, = \,h^{N+2}|A_N''| . 
\]

We also define the Wronskian-like function:
\[
  W_{h,N}\,=\,(u_{h,N}^{+})' u_{h,N}^- -{u_{h,N}^{+}} (u_{h,N}^-)'.
\]
  
The following lemma records the needed estimates.

\begin{Lemma}
\label{lem:ubds}
  Fix $\eps \in (0,1)$, and set $a_h= h^{1-\eps}$. For any $N$, there exists constants $C$ such that 
  \begin{equation*}
    \begin{split}
      \| u_{h,N}^\pm - A_0\exp\left( \frac{\pm iS}{h} \right) \|_{C^0([a_h,b])} \,&\leq\, Ch^\eps, \\  
      \| h (u_{h,N}^\pm)' - iA_0^{-1}\exp\left( \frac{ \pm iS}{h} \right) \|_{C^0([a_h,b])} \,&\leq\, Ch^\eps.
      \end{split}
      \end{equation*}
      
      In addition, we have
      \begin{equation*}
          \begin{split}
      \|r_{h,N}^{\pm}\|_{C^0([a_h,b])}\,\leq\, Ch^{\eps{N+2}}, & \ \ \| h^2W_{h,N}'\|_{C^0([a_h,b])}\,\leq \, Ch^{\eps(N+2)},\\
      \| h^2W_{h,N} - 2ih\|_{C^0([a_h,b])}\,\,&\leq Ch^{\eps(N+2)}.
    \end{split}
  \end{equation*} 
\end{Lemma}

\begin{proof}
  Using the estimates of the preceding section, we have
  \[
    \left| u^+_{h,N}(x) - A_0(x)\exp\left( \frac{iS(x)}{h} \right)\right | \,\leq \, C\sum_{k=1}^{N} h^k (1+x^{\gamma-k}).
  \]
  The first estimate follows since
  \[
    \forall x\in [a_h,b],~~\left|\frac{h}{x}\right| \leq \frac{h}{a_h}\,\leq h^{\eps},
  \]
  and $h^\eps \gg h$. 
  The following three estimates are obtained in a similar way, starting from Proposition \ref{prop:estAk}.
  The last one follows from integrating the third estimate on $[a_h,b]$ and observing that
  $h^2W_{h,N}(b)\,=\,2ih$.
\end{proof}

We use these estimates to prove the following, in which it is convenient to introduce the semiclassical $C^1$ norm defined by
\[
  \| v\|_{C^1(I)} \,=\,\max \{ |v(x)|, |hv'(x)|,~~x\in I \}. 
\]

\begin{Theorem}\label{thm:extsol}
  For any $\eps\in (0,1)$ and any $D$, there exists $N,C$ and two independent solutions $\psie{\pm}$ such that,
  \[
    \| \psie{\pm} - u_{h,N}^\pm \|_{C^1([h^{1-\eps},b])} \,\leq C\, h^{D}.
  \]
\end{Theorem}

\begin{proof}
  Starting from \eqref{eq:varconst}, we have that $(B_+,B_-)$ satisfies
  \begin{equation}\label{eq:systext}
    \left \{
      \begin{array}{c}
        h^2B_+' (u_{h,N}^+)'\,+\,h^2B_-' (u_{h,N}^-)'=B_+r_{h,N}^+\,+\,B_-r_{h,N}^-, \\
        B'_+u_{h,N}^+\,+\,B_-'u_{h,N}^-=0.\\
      \end{array}
    \right .
  \end{equation}
  For $\tl=\pm$ and $\tr=\pm$, we define the integral operator $H^{\tl ,\tr}$ on $C^0([a_h,b])$ by
  \[
    H^{\tl ,\tr}[B](x)\,=\,\int_{x}^b \frac{r_{h,N}^\tl(\xi)u_{h,N}^\tr(\xi)}{h^2W_{h,N}(\xi)}B(\xi)\,d\xi.
  \]
 We also define the matrix operator $\Hbb$ acting on  $\left(C^0([a_h,b])\right)^2$ by
  \[
    \mathbb{H} \,=\,
    \begin{pmatrix}
      H^{+,+} & H^{+,-} \\
      H^{-,+} & H^{-,-}
    \end{pmatrix}.
  \]
  If the system \eqref{eq:systext} can be inverted, we can express $B_{\pm}'$ depending on $B_{\pm}$.
  By integration, we obtain that there exist two constants $\beta_+$ and $\beta_-$ such that
  \[
    \begin{pmatrix}
      B_+ \\
      B_-
    \end{pmatrix}
    \,=\,
    \begin{pmatrix}
      \beta_+ \un \\
      \beta_- \un
      \end{pmatrix}
      \,+\,
      \Hbb
      \begin{pmatrix}
         B_+ \\
      B_-
      \end{pmatrix}.
    \]
  
    Using Lemma \ref{lem:ubds}, there exists a constant such that
    \[
      \forall \tl,\tr = \pm,~~ \| H^{\tl, \tr}\|_{\Lcal(C^0([a_h,b]))}\,\leq\, Ch^{\eps(N+2)-1}.
    \]
    It follows that $\id -\Hbb$ is invertible and that, for any choice of $(\beta_+,\beta_-)$
    \[
      \| \psi-\beta_+u_{h,N}^+,-\beta_-u_{h,N}^- \|_{C^1([a_h,b])}\,\leq\,Ch^{\eps(N+2)-1}(|\beta_+|\,+\,\beta_-|). 
    \]
    The claim follows by choosing first $N$ so that $\eps(N+2)-1>D$ and then
    $(\beta_-,\beta_+)= (1,0)$ and $(0,1)$.   
\end{proof}

\section{The interior region}
\label{sec:int}

The WKB expansion gives us a good approximation for the true solutions on intervals
of the form $[h^{1-\eps},b]$ for any $\eps$. We now show that there exists $\eps_0>0$ such that
on the interval $[0,h^{1-\eps_0}]$ we get a good approximation by comparing the
solution to trigonometric solutions. Since we can choose $\eps$ arbitrarily small, matching will
be possible on the interval $[h^{1-\eps},h^{1-\eps_0}]$. 

It is convenient to make the change of
independent variables by setting $z= \frac{\sqrt{E}}{h}x$. So that we look for
solutions to the following equation:
\begin{equation}\label{eq:rescaled}
  \ddot{v}\,+\,v\,=\,h^{\gamma}z^{\gamma}\tilde{W}(hz)v(z),
\end{equation}
where $\tilde{W}(\cdot)\,=\, E^{-(1+\gamma/2)} W(\frac{\cdot}{\sqrt{E}})$.
We study this equation on $[0,b_h]$ where $b_h$ will eventually be $h^{-\delta}$.
We assume that $\delta<1$ so that $b_h=O(h^{-1})$.
This ensures that for any $k$, $z\mapsto \tilde{W}^{(k)}(hz)$ is uniformly bounded
on $[0,b_h]$.

\begin{rem}
  Despite the rescaling, we still denote by $[0,b_h]$ the interval we are working on. Observe that
  $b_h= h^{-\delta}$ corresponds to an interval of order $h^{1-\delta}$ in the original setting.
\end{rem}

We define the following integral operators on $C^0([0,b_h])$:
\[
  \begin{split}
    K[v](z)&=\, \int_0^{z} \sin(z-\zeta) h^\gamma \zeta^\gamma \tilde{W}(h\zeta)v(\zeta)\,d \zeta, \\
    \\
    K'[v](z)&\,=\, \int_0^{z} \cos(z-\zeta) h^\gamma \zeta^\gamma \tilde{W}(h\zeta)v(\zeta)\,d \zeta.
   \end{split}           
 \]
By differentiation, we have $(K[v])'=K'[v]$.  
A straightforward computation shows that any solution $v$ to \eqref{eq:rescaled} can be written
\[
  v\,=\,a_+ e_+ \,+\, a_-e_- \,+\,K[v], 
\]
where we have defined $e_{\pm}(z)\,=\,e^{\pm iz}$.

When $(I-K)$ is invertible, we obtain a basis of solutions to \eqref{eq:rescaled} by computing
\[
  (I-K)^{-1} e_{\pm}.
\]

\begin{Lemma}
  For any $\delta \in (0,\gamma)$, set $b_h\,=\,h^{-\frac{\gamma-\delta}{\gamma+1}}$. There exists $C$ such that
  \[
    \| K\|_{\Lcal(C^0([0,b_h])}\,\leq Ch^{\delta}.
  \]
  the constant $C$ is uniform for $(W,E)$ in $\Kcal \times K$. 
  For $h$ small enough, the operator $(I-K)$ is an invertible endomorphism 
  of $C^0([0,b_h])$. 
\end{Lemma}
\begin{proof}
  Fix $\delta \in (0,\gamma)$ and set $b_h$ as in the lemma. We observe that
  $b_h\,=\,O(h^{-1})$ and that $h^\gamma b_h^{\gamma+1} = O(h^\delta).$
  Since $\tilde{W}$ is continuous, we have 
  \[
   \forall v\in C^0([0,b_h])\,~ \| K[v]\|_{C^0([0,b_h])}\,\leq\, C(\tilde{W}) h^\gamma b_h^{\gamma+1} \|v\|_{C^0([0,b_h])}
 \]
 with $C(\tilde{W})\,=\,\sup \{ |\tilde{W}(z)|,~z\in [0,b_h] \}$.
  The claim follows.
\end{proof}

For the rest of this section, we fix some $\delta \in (0,\gamma)$ and set $\eps_0 = \frac{\gamma-\delta}{\gamma+1}$.
Using a Taylor expansion for $W$ near $0$, we can write, for any $N$
\[
  \tilde{W}(hz)\,=\,\sum_{j=0}^{N-1} h^jw_j z^j \,+\,h^Nz^Nr_{N}(hz), 
\]
where the $w_j$ are smooth functions of $E$ on $K$.

We thus define the operators $L_j$ on $C^0([0,b_h])$ by
\[
   L_j[v](z)\,=\, w_j\int_0^{z} \sin(z-\zeta) h^{\gamma+j} \zeta^{\gamma+j} v(\zeta)\,d \zeta.
\]
We will also need the operators $L_j'$ so that $L_j'[v] =(L_j[v])'$. The kernel of $L_j'$ is obtained by changing the
sine function in the kernel of $L_j$ to the cosine function.

A straightforward computation shows that
\[
  \forall N,~~K\,=\,\sum_{j=0}^{N-1} L_j \,+\,R_N,
\]
with
\[
 R_N[v](z)\,=\, \int_0^{z} \sin(z-\zeta) h^{\gamma+N} \zeta^{\gamma+N} r_N(h\zeta) v(\zeta)\,d \zeta. 
\]
By the same procedure as above, we estimate
\[
  \begin{split}
    \|R_N\|_{\Lcal(C^0[0,b_h])}&\,\leq\, C_Nh^{\gamma+N}b_h^{\gamma+N+1} \\
                               & \,\leq\,C_N h^{\delta_N},\\
  \end{split}
\]
with $\delta_N\,=\,\gamma+N +\frac{\delta-\gamma}{\gamma+1}(\gamma+N+1)$ and a constant $C_N$ that is uniform for
$(W,E) \in \Kcal\times K$.
We see that
\[
  \delta_N\,=\,N(1+\frac{\delta-\gamma}{\gamma+1})+\delta \geq N\frac{1+\delta}{1+\gamma}.
\]

In the same way, we also get that
\[
  \forall j, ~~\|L_j\|_{\Lcal(C^0[0,b_h])} \,\leq\, C_jh^{\delta_j},~~\text{with}~~\delta_j\,=\,\delta \,+\,j \frac{1+\delta}{1+\gamma}.
\]
All these estimates have been obtained by crudely bounding $|\sin (z-\zeta)|$ by $1$. It follows that the same estimates apply to $K', L'$
and $R_N'$ (with the now standard definition for the latter).

Finally, setting $d= \frac{1+\delta}{1+\gamma}$, we have the following norm estimates, for all $j$, $N$ and small enough $h$ , 
\begin{equation}\label{normest}
  \begin{split}
    \max \left( \| K\|_{\Lcal(C^0[0,b_h])},\| K'\|_{\Lcal(C^0[0,b_h])} \right)&\,\leq\, C h^\delta ,\\
    \max \left(\| L_j\|_{\Lcal(C^0[0,b_h])},\| L_j'\|_{\Lcal(C^0[0,b_h])}\right) &\,\leq\, C_j h^{jd},\\
    \max \left( \| R_N\|_{\Lcal(C^0[0,b_h])},\| R_N'\|_{\Lcal(C^0[0,b_h])}\right) &\,\leq\, \hat{C}_N h^{Nd},
  \end{split}
\end{equation}
uniformly for $(W,E)\in \Kcal\times K$.

Recall that, for two operators $A$ and $B$ acting on a Banach space $X$, we have
\[
  \|A\|_{\Lcal(X)},~\|B\|_{\Lcal(X)}\leq \frac{1}{2} \implies \| (\id -A)^{-1}-(\id-B)^{-1}\|_{\Lcal(X)}\,\leq\, 4\|A-B\|_{\Lcal(X)}. 
\]

It follows that
\[
  \forall D,~\exists N,~~\|(\id -K)^{-1} - (\id -\sum_{j=0}^{N-1}L_j)^{-1}\|_{\Lcal(C^0([0,b_h])}\,\leq\,4C_N h^{Nd}\,\leq h^D.
\]

Using a Neumann series expansion, we have 
\[
  (\id-\sum_{j=0}^{N-1}L_j)^{-1}\,=\,\sum_{n=0}^{\infty} \sum_{(j_0,j_1,\cdots j_{n}) \in \{ 0 ,\dots, N-1 \}^n} L_{j_n}\cdots L_{j_0}.
\]
It is convenient to write $\vec{j}=(j_0,\cdots,j_n)$ and
\[
  L_{\vec{j}}\,=\,L_{j_n}\cdots L_{j_0},
\]
so that the infinite sum can be seen as a sum over all tuples 
of arbitrary length.

We also define $L_{\vec{j}}'$ so that $L_{\vec{j}}'[v] = (L_{\vec{j}}[v])'$. By definition,
we see that $L_{\vec{j}}'$ is obtained by replacing the final $L_{j_n}$ by $L_{j_n}'$.  Using the preceding estimates \eqref{normest}, we compute
\[
  \max\left( \|L_{\vec{j}}\|_{\Lcal(C^0([0,b_h])},\|L_{\vec{j}}'\|_{\Lcal(C^0([0,b_h])}\right) \,\leq C\, h^{(n+1)\delta\,+\,(\sum_{0}^n j_i)d}. 
\]
For any $D$ and any $j$, we define $n_j$ to the greatest integer $n$ such that
\[
  (n+1)(\delta +jd)< D,
\]
and we define $\Jcal(D)$ to be the set of tuples such that, for any $j$,
the index $j$ occurs at most $n_j$ times in $\vec{j}$ and we set
\[
  T_D\,=\,\sum_{\vec{j}\in \Jcal_D} L_{\vec{j}}.
\]

We claim that
\[
  \|(\id-\sum_{j=0}^{N-1}L_j)^{-1} -T_D\|_{\Lcal(C^0([0,b_h]))}\leq C h^D.
\]
Indeed, we have
\[
  \|(\id-\sum_{j=0}^{N-1}L_j)^{-1} -T_D\|_{\Lcal(C^0([0,a_h]))} \,\leq\,\sum_{\vec{j}\notin \Jcal_D} \|L_{j_n}\|\cdots\|L_{j_0}\|.
\]

In the product, at least one index $j$ occurs more than $n_j$ times. Fixing an index $j_0$, the sum of all terms in which there are
more that $n_{j_0}$ factors $\|L_{j_0}\|$ is bounded above by 
\[
  \left(\sum_{n=0}^\infty \|L_{0}\|^n\right)\left(\sum_{n=0}^\infty \|L_{1}\|^n\right)\cdots
  \left(\sum_{n=n_{j_0}+1}^\infty \|L_{j_0}\|^n\right)\cdots\left(\sum_{n=0}^\infty \|L_{{N-1}}\|^n\right)\leq C h^D
\]
by definition of $n_{j_0}$. The claim follows by considering all possible $j_0$.   

Finally, we obtain 

\begin{prop}\label{prop:TD}
  For any $D$, there exists two independent solutions
  $\psi^{\pm}$ such that, for $h$ small enough 
  \[
    \| \psi^\pm - T_D e^\pm \|_{C^1([0,b_h])} \,\leq\, C h^D,
  \]
  and the constant $C$ is uniform for $(W,E)\in \Kcal\times K$.
\end{prop}
\begin{proof}
  Setting $\psi^\pm=(I-K)^{-1}e^{\pm}$ yields two solutions. The preceding estimates give that
  \[
    \|\psi^{\pm}-T_De^{\pm}\|_{C^0([0,b_h])}\,\leq Ch^D. 
  \]
  But the same proofs also yield that
  \[
    \|(\psi^{\pm})'-(T_De^{\pm})'\|_{C^0([0,b_h])}\,\leq Ch^D 
  \]
  so that the claim follows.
  The fact that the two solutions are independent is obtained by computing the Wronskian at $0$. The latter does not vanish since
  the Wronskian of $e^\pm$ does not vanish and the error made is of order $h^D$.
\end{proof}

\subsection{Asymptotic behavior of $\psi^\pm$}
As in the case for the solutions in the exterior region, we
first introduce the convenient setting for our asymptotic expansions.

For $\alp \in \R \setminus \mathbb{Q}$, we define the set $\Fcal_\alp$ of smooth functions on $[0,+\infty[$
that admit the following asymptotic expansion near $\infty$ :
\[
  f(z)\sim c_\alpha\,+\,\sum_{\ell \geq 0} a_\ell z^{\alpha -\ell}.
\]
For $f$ in $\Fcal_\alpha$, there exists $C$ such that
\[
  \forall z\in [0,+\infty[,~~|f(z) | \,\leq\,C(1+z^{\alpha}).
\]

\begin{rem}
    To allow $\gamma \in \mathbb{Q}$, we would have to include $z^{\alpha-\ell} \log z $ terms in the asymptotic expansion here also.
\end{rem}

A sequence $F_h$ of smooth functions on $[0,b_h]$ is said to be admissible if there exists
a sequence of functions $f_{m,n}$ such that
$f_{m,n} \in \Fcal_{m\gamma+n}$ such and
\[
  F_h(z)\,\sim\,\sum_{m\geq 1,n\geq 0} h^{m\gamma+n}f_{m,n}(z)
\]
in the following sense:
\[
  \forall D,~\exists M,N,~\| F_h -\sum_{\substack{1\leq  m\leq M\\0\leq n\leq N}} h^{m\gamma+n}f_{m,n}\|_{C^1([0,b_h])} \leq Ch^D.  
\]

We observe that this definition is legitimate because if $F_h$ is admissible, we can define the
$f_{m,n}$ recursively by a limiting procedure. 

The main result of this section is then the following:
\begin{Theorem}\label{thm:intsol}
  Fix $\delta\in (0,\gamma)$ and $b_h\,=\,h^{-\frac{\gamma-\delta}{\gamma+1}}$.
  Let $\psi^{+}$ be the solution constructed in the previous section then there exist admissible functions
  $F^{+,\pm}_h$ such that
  \[
    \psi^{+}\,=\,F^{+,+}e_+\,+\,F^{+,-}e_-.
  \]
  A similar statement holds for $\psi^-$ with admissible functions $F^{-,\pm}$.
\end{Theorem}
\begin{proof}
  By definition of admissibility and using Proposition \ref{prop:TD}, it suffices to show that $T_De^+$ is admissible.
  Since $T_D$ is a finite sum of operators $L_{\vec{j}}$, it suffices to study $L_{\vec j}(e^+)$. Fix $\vec{j}=(j_0,\cdots, j_m)$,
  and set $n\,=\,\sum\limits_{i=0}^n j_i$.
  Using Lemma \ref{lem:L_jF} below, and a straightforward induction, we have that
  \[
    L_{j_n}L_{j_{n-1}}\cdots L_{j_0} (e^+)\,=\,h^{(m+1)\gamma+n}\left(f_{\vec{j}}^+ e^+\,+\,f_{\vec{j}}^- e^-\right) 
  \]
  with
  \[
    f_{\vec{j}}^\pm \in \Fcal_{(n+1)(\gamma+1)\,+\,\sum_{i=0}^N j_i\,}.
  \]
\end{proof}

\begin{Lemma}\label{lem:L_jF}
  For any $j$, and any $\alpha$.If $f,g \in \Fcal_{\alpha}$, there exists $f_+,g_+$ and $f_-,g_-$ in $\Fcal_{\alpha+\gamma+j+1}$
  such that
  \[
    L_{j}[fe^+]\,=\,h^{\gamma+j}(f_+e^+\,+\,f_-e^-),~~ L_{j}[ge^-]\,=\,h^{\gamma+j}(g_+e^+\,+\,g_-e^-).
  \]
\end{Lemma}
\begin{proof}
  Using complex conjugation, it suffices to obtain asymptotic expansions for the following expressions :
  \[
    \int_0^z \zeta^{\gamma+j}f(\zeta) \, d\zeta,~~e^{-2iz} \int_0^z e^{2i\zeta}\zeta^{\gamma+j}f(\zeta)d\zeta.
  \]
  For the former, replacing $f$ by each term in its asymptotic expansion, we obtain the result directly.
  For the latter, we also replace $f$ by each term in its asymptotic expansion. We obtain an asymptotic expansion
  for the resulting integral using repeated integration by parts and the fact that the power is never $-1$.
\end{proof}

\section{Matching}
\label{sec:matching}

We recall that for $\delta\in (0,\gamma)$, we define $\eps= \frac{\gamma-\delta}{\gamma+1}$.
We choose $\eps_0<\eps_1$ in this range and observe that our construction for
the interior solutions is valid on $[0,h^{1-\eps_1}]$ and the WKB construction is valid on $[h^{1-\eps_0},b]$. 
In the sequel the $\eps$ will be in $[\eps_0,\eps_1]$ and we will set $I\,=\,[h^{1-\eps_0},h^{1-\eps_1}]$. 

\subsection{Exponent set for the interior solutions}
Since $I\subset [0,h^{1-\eps_0}]$, according to Theorem \ref{thm:intsol}, there exist two independent solutions $\psii{\pm}$ that admit the following asymptotic expansion :

\[
  \begin{split}
    \psii{+}(x)\,\sim\,&\exp\left(i\frac{x\sqrt{E}}{h}\right)\left[ 1\,
                         +\,\sum_{m\geq 1, n\geq 0} h^{m\gamma+n}\left( c_{m,n}^{+,+}+f_{m,n}^{+,+}\left( \frac{x\sqrt{E}}{h}\right) \right)\right] \\
                       & + \exp\left(-i\frac{x\sqrt{E}}{h}\right)\,
                         \sum_{m\geq 1, n\geq 0} h^{m\gamma+n} \left( c_{m,n}^{+,-}\,+\,f_{m,n}^{+,-}\left( \frac{x\sqrt{E}}{h} \right) \right), \\
    \psii{-}(x)\,\sim\,& \exp\left(i\frac{x\sqrt{E}}{h}\right)\,
                         \sum_{m\geq 1, n\geq 0} h^{m\gamma+n}\left(c_{m,n}^{-,+}\,+\,f_{m,n}^{-,+}\left(\frac{x\sqrt{E}}{h}\right)\right)\\
                       & + \exp\left(-i\frac{x\sqrt{E}}{h}\right)\left[ 1\,+\,
                         \sum_{m\geq 1, n\geq 0} h^{m\gamma+n}\left(c_{m,n}^{-,-}\,+f_{m,n}^{-,-}\left(\frac{x\sqrt{E}}{h} \right) \right)\right].\\
  \end{split}
\]
Each term in the asymptotic expansions can be written as
\[
  h^{m\gamma+n}~~\text{or}~~h^{m\gamma+n}\left(\frac{x}{h}\right)^{m\gamma +n -\ell}.
\]

We obtain the following proposition

\begin{prop}\label{prop:expsetint}
  For $\eps \in (\eps_0,\eps_1)$, there exist asymptotic expansions ${\fpip_h, \fpim_h}$, ${\fqip_h, \fqim_h}$ such that
  \[
    \begin{split}
      \psii{+}(h^{1-\eps})\,=\,& \exp(ih^{-\eps}\sqrt{E})(1\,+\,\fpip_h) \,+\,\exp(-ih^{-\eps}\sqrt{E})\fqip_h, \\
      \psii{-}(h^{1-\eps})\,=\,& \exp(ih^{-\eps}\sqrt{E})\fpim_h \,+\,\exp(-ih^{-\eps}\sqrt{E})(1+\fqim_h).
    \end{split}  
  \]
  Moreover, the exponent set for these asymptotic expansions is
  \[
    \Acal^{\mathrm{int}} =\big\{ m\gamma+n,~m\geq 1,n\geq 0\big\} \cup
    \big \{ \eps \ell\,+\,(1-\eps)(m\gamma+n), \ell \geq 0, m\geq 1, n\geq 0\big \}.
  \]  
\end{prop}
\begin{proof}
When evaluating the preceding asymptotic expansions at $x_h \defeq h^{1-\eps}$, we obtain terms corresponding
to $h^{m\gamma+n}$ that yield the first part of the exponent set. The second set comes from the terms
\[
  h^{m\gamma +n}\left(\frac{x_h}{h}\right)^{m\gamma+n-\ell}\,=\,h^{\eps \ell}h^{(1-\eps)(m\gamma+n)}.
\]
\end{proof}

\subsection{Exponent set for the WKB solutions}
Since $I\subset [h^{1-\eps_0},b]$, the WKB construction yields 
two independent solutions to the Schrödinger equation $\psie{\pm}$ such that

\begin{equation}\label{eq:asympext}
  \begin{split}
  \psie{+}(x)\,\sim\,&\exp\left( \frac{iS(x)}{h} \right)\left [ q(x)^{-\frac{1}{4}}\,+\,\sum_{k\geq 1} h^k A_k^+(x)\right ],\\
  \psie{-}(x)\,\sim\,&\exp\left( \frac{-iS(x)}{h} \right)\left [ q(x)^{-\frac{1}{4}}\,+\,\sum_{k\geq 1} h^k A_k^-(x)\right ].
  \end{split}
\end{equation}

We now proceed to get asymptotic expansions when we evaluate at $x_h=h^{1-\eps}$.
From Theorem \ref{thm:extsol}, we get
\[
  \begin{split}
     q(x)^{- \frac{1}{4}}\,\sim\,& E^{-\frac{1}{4}} \,+\,\sum_{m\geq 1,n\geq 0} a_{0,m,n} x^{m\gamma+n},\\
    h^kA_k(x) \,\sim\,~& h^kc_k\,+\,\left(\frac{h}{x}\right)^k\sum_{m\geq 1,n\geq 0} a_{k,m,n}x^{m\gamma +n}.\\
  \end{split}
\]
Evaluating at $x_h$ we obtain that each expression inside brackets in \eqref{eq:asympext} admits an asymptotic
expansion with exponent set
\[
  \{ k\eps+(1-\eps)(m\gamma+n) k\geq 0, m\geq 1, n\geq 0\} \cup \N.
\]

We now study the prefactor by computing 
\[
  \begin{split}
    S(x) =& -\int_{x}^b \sqrt{E-V(y)}dy \\
          &= -\int_0^b\sqrt{E-V(y)}\,dy\,\,+\,\,x\sqrt{E} \,+\,\int_{0}^x \left[ \sqrt{E-V(y)}-\sqrt{E}\right]\,dy.
  \end{split}
\]
Setting
\[
  \sigma_E\,=\,\int_0^b\sqrt{E-V(y)}\,dy~\text{and}~T_E(x)\,=\,\int_{0}^x \left[ \sqrt{E-V(y)}-\sqrt{E}\right]\,dy,
\]
we observe that $T_E$ has a generalized Taylor expansion with exponent set
\[
  \{ m\gamma+n+1,~m\geq 1, n\geq 0\}.
\]
  
It follows that there exist coefficients $t_{m,n}$ such that
\[
  \frac{i}{h}T_E(h^{1-\eps}) \,=\,h^{(1-\eps)(\gamma+1)-1}\sum_{m\geq 0,n\geq 0} t_{m,n}h^{(1-\eps)(m\gamma+n)}.
\]
Recalling that $\eps\,=\,\frac{\gamma-\delta}{\gamma+1}$, we observe that
\[
  (1-\eps)(\gamma+1)-1\,=\,\delta >0,
\]
so that the preceding expression is $O(h^\delta)$.
More precisely, we obtain that
\[
  \exp\left(-\frac{i}{h}T_E(h^{1-\eps})\right)\,=\,1\,+\,t_h,
\]
where $t_h$ admits an asymptotic expansion with exponent set
\[
  \{ \ell\delta + m(1-\eps)\gamma + n(1-\eps),~\ell \geq 1, m\geq 0, n\geq 0\}.
\]

We obtain the following.

\begin{prop}\label{prop:expsetext}
  For $\eps \in (\eps_0,\eps_1)$, there exist asymptotic expansions denoted $\fpep_h, \fpem_h,$ such that
  \[
    \begin{split}
      \psie{+}(h^{1-\eps})\,=\,& \exp\left(-\frac{\sigma_E}{h}\right)\exp(ih^{-\eps}\sqrt{E}) E^{-\frac{1}{4}}
                                 (1\,+\,\fpep_h)  ,\\
      \psie{-}(h^{1-\eps})\,=\,& \exp\left(\frac{\sigma_E}{h}\right)\exp(-ih^{-\eps}\sqrt{E}) E^{-\frac{1}{4}}(1\,+\,\fpem_h).\\
    \end{split}  
  \]
  The exponent set for these asymptotic expansions is
  \[
    \Acal^{\mathrm{ext}} \,=\,\big \{ k\eps\,+\,\ell\delta\,+\,m(1-\eps)\gamma\,+\,n(1-\eps), k\geq 0, \ell\geq 1, m\geq 0, n\geq 0 \big\}\cup \N.
  \]  
\end{prop}
\begin{proof}
  The exponent set comes from studying the product
  \[
    (1+t_h)\left(c_kh^k\,+\,h^{k\eps}\sum_{m\geq 1, n\geq 0} a_{k,m,n} h^{(1-\eps)(m\gamma+n)}\right)
  \]
  for each $k$.
  Since the exponent set that we obtain is discrete, the claim follows.
\end{proof}

\subsection{Proof of Theorem \ref{thm:matching}}
Let $\phi$ be a solution to the Schrödinger equation.
In the interior region, we can write
\[
    \phi \,=\, \api \psii{+}\,+\,\ami \psii{-} 
\]
and in the exterior region, we write 
\[
    \phi \,=\, \ape \psie{+}\,+\,\ame \psie{-}. 
\]
Using Theorem \ref{thm:extsol} we have
\[
\begin{split}
\psie{\pm}(b) & = \, (E-V(b))^{-\frac{1}{4}}\\ 
h(\psie{\pm})' (b) & = \pm i (E-V(b))^{\frac14}\left( 1\,+\,hg^{\pm}_h\right),
\end{split}
\]
where $g^\pm_h$ has an asymptotic expansion in integral powers of $h$. It follows that 
\begin{equation}\label{eq:evalb}
  \begin{pmatrix}
    (E-V(b))^{\frac{1}{4}}\ \phi(b) \\
    (E-V(b))^{-\frac{1}{4}}\ h\phi'(b)
  \end{pmatrix}
  \,=\,
  \begin{pmatrix}
    1 & 1\\
    i & -i
  \end{pmatrix}
  \begin{pmatrix}
  1 + \frac{h}{2}g^+_h & -\frac{h}{2} g^-_h \\
  -\frac{h}{2}g^+_h & 1+\frac{h}{2}g^-_h 
  \end{pmatrix}
  \begin{pmatrix}
    \ape\\
    \ame
  \end{pmatrix}.
\end{equation}

In the matching region, both expressions for $\phi$ are valid. Thus, evaluating at
$h^{1-\eps}$ we obtain
\[
\begin{split}
    & \exp(ih^{-\eps}\sqrt{E})\big( \api(1\,+\,\fpip_h) + \ami \fpim_h\big)\\
    &  \hspace{1cm} \,+\, \exp(ih^{\eps}\sqrt{E})\big( \api \fqip_h \,+\,\ami (1\,+\,\fqim_h) \big)\\
    & = \exp(ih^{-\eps}\sqrt{E})\ape\exp\left(\frac{-i\sigma_E}{h}\right)E^{-\frac{1}{4}}\big(1\, +\,\fpep_h\big),  \\
    & \hspace{1cm} \,+\,
    \exp(-ih^{-\eps}\sqrt{E})\ame\exp\left(\frac{i\sigma_E}{h}\right)E^{-\frac{1}{4}}\big(1\,+\,\fpem_h\big).
  \end{split}
\]  
We multiply by $E^{\frac{1}{4}}$, and observe that we can pairwise identify the asymptotic expansions in front
of $\exp(\pm ih^{-\eps}\sqrt{E})$ yielding the following system
of equations:
\begin{multline*}
  \begin{pmatrix}
    1\,+\, \fpip_h & \fpim_h \\
    \fqip_h & 1+\fqim_h
  \end{pmatrix}
  \begin{pmatrix}
    E^{\frac{1}{4}}\ \api\\
    E^{\frac{1}{4}}\ \ami
  \end{pmatrix}\\
  \,=\,
  \begin{pmatrix}
    (1 + \fpep_h)\exp\left(\frac{-i\sigma_E}{h}\right) & 0 \\
    0 & (1+\fpem_h)\exp\left(\frac{+i\sigma_E}{h}\right)
    \end{pmatrix}
    \begin{pmatrix}
      \ape\\
      \ame
    \end{pmatrix}.
  \end{multline*}
    
Then, we observe that 
\[
  \begin{pmatrix}
    E^{\frac{1}{4}}\ \phi(0) \\
    E^{-\frac{1}{4}}\ h\phi'(0)
  \end{pmatrix}
  \,=\,
  \begin{pmatrix}
    1 & 1\\
    i & -i
  \end{pmatrix}
  \begin{pmatrix}
    E^{\frac{1}{4}}\ \api\\
    E^{\frac{1}{4}}\ \ami
  \end{pmatrix}
\]
so that, using \eqref{eq:evalb}, we obtain the following relation
\[
  \begin{split}
  \begin{pmatrix}
    E^{\frac{1}{4}}\ \phi(0) \\
    E^{-\frac{1}{4}}\ h\phi'(0)
  \end{pmatrix}
  &\,=\,
 \begin{pmatrix}
    1 & 1\\
    i & -i
 \end{pmatrix}
 \begin{pmatrix}
    1\,+\, \fpip_h& \fpim_h \\
    \fqip_h & 1+\fqim_h
 \end{pmatrix}^{-1}\\
 &   
 \cdot \begin{pmatrix}
    (1 + \fpep_h)\exp(\frac{-i\sigma_E}{h}) & 0 \\
    0 & (1+\fpem_h)\exp(\frac{+i\sigma_E}{h})
 \end{pmatrix}\\
 &\cdot \begin{pmatrix}
  1 + \frac{h}{2}g^+_h & -\frac{h}{2} g^-_h \\
  -\frac{h}{2}g^+_h & 1+\frac{h}{2}g^-_h 
  \end{pmatrix}^{-1}\begin{pmatrix}
    1 & 1\\
    i & -i
         \end{pmatrix}^{-1}
   \begin{pmatrix}
    (E-V(b))^{\frac{1}{4}}\ \phi(b) \\
    (E-V(b))^{-\frac{1}{4}}\ h\phi'(b)
  \end{pmatrix}.
\end{split}   
\]

We can thus rewrite
\[
  \begin{pmatrix}
    E^{\frac{1}{4}}\ \phi(0) \\
    E^{-\frac{1}{4}}\ h\phi'(0)
  \end{pmatrix}
    \,=\,
    \mathbb{M}_h(E)\,
    \begin{pmatrix}
    (E-V(b))^{\frac{1}{4}}\ \phi(b) \\
    (E-V(b))^{-\frac{1}{4}}\ h\phi'(b)
  \end{pmatrix}.
\]

We compute the matrix $\Mbb_h(E)$ to obtain
\[
    \Mbb_h(E)
    \,=\,
    \begin{pmatrix}
      \cos \frac{\sigma_E}{h} & -\sin \frac{\sigma_E}{h} \\
      \sin\frac{\sigma_E}{h} & \cos \frac{\sigma_E}{h}
    \end{pmatrix}
    \,+\cos \frac{\sigma_E}{h} \Abb_{h}^+(E)\,+\, \sin\frac{\sigma_E}{h}\Abb_{h}^-(E)
\]
where the matrices $\Abb_h(E)$ admit asymptotic expansions in $h$ with smooth coefficients in $E$.
By construction, the exponent set of the latter is contained in the sum $\Acal^{\mathrm{int}}\,+\,\Acal^{\mathrm{ext}}$.
But the exponent set should not depend on $\eps$. It follows that only powers that can be written
$m\gamma +n$ with $m\geq 1$ and $n\geq 0$ can have a non-zero coefficient, which concludes the proof.

\begin{rem}
For any fixed $W\in \Kcal$, the coefficients in the asymptotic expansions 
  are in principle computable and shown to be smooth for $E\in K$. Moreover, 
  the expansion is uniform for $(W,E)\in \Kcal\times K$.    
\end{rem} 
\section{Examples}
\label{sec:examples}
    \subsection{Singular Bohr-Sommerfeld rules}
    We apply the main theorem to obtain a singular Bohr Sommerfeld rule for the problem with Dirichlet
    boundary condition at $0$ and $b$.

    The eigenvalue equation in that case can be written :
    \[
    \Mbb_h(E)\,
    \begin{pmatrix}
      0 \\ 1
    \end{pmatrix}
    \,=\,
    \begin{pmatrix}
      0 \\ *
    \end{pmatrix}.
    \]
    It follows the following equation :
    \[
    \sin(\frac{\sigma_E}{h})\left(1\,+\,\mathfrak{a}^-_h(E)\right)\,+\,\cos(\frac{\sigma_E}{h})\mathfrak{a}^+_h(E)\,=\,0.
    \]
    where $\mathfrak{a}_{h}^\pm$ have asymptotic expansions with exponent set   given by 
    $${\{ m\gamma+n,~m\geq 0,n\geq 0 \}\setminus\{0\}}.$$

    This is easily transformed into the following Bohr-Sommerfeld rules.
    \[
    \frac{1}{h} \int_0^b (E-V(y))^{\frac{1}{2}}\,dy \,+\,\sum_{(m,n)\neq  (0,0)} h^{m\gamma +n} \theta_{m,n}(E)\,=\,k\pi.
    \]
    This is the standard form with two modifications: the integrand is not smooth at $0$, and the exponent set is not
    the integers.

    \subsection{Singular Bohr-Sommerfeld rules on the half-line}
    In this section, we consider a Schrödinger operator on the half-line $[0,+\infty)$.
    We assume that the potential $V$ is smooth on $(0,+\infty)$ and satifies the same assumptions as before on some interval $(0,c]$
    and we assume, for simplicity, that $V$ is increasing on $[c,+\infty)$. We also assume that the energy window $K$
    is such that $V^{-1}(K) \subset (0,c]$.

    \begin{rem}
      We could consider the more general setting described in \cite{hillairet2023eigenvalue}.
    \end{rem}

    For $E\in K$, let $G_h(\cdot ;E)$ be a $L^2$ solution to the eigenvalue equation on the half-line.
    By assumption, such $L^2$ solutions form a one-dimensional vector space. Fix $b<c$, such that $V^{-1}(K) \subset (b,c)$
    the analysis of the present paper gives a good understanding on $[0,b]$.

    On the interval $[b,+\infty)$, the potential is smooth so that the usual theory applies. In particular, we can use the Maslov Ansatz and look for $G_h$ as an oscillatory integral of the following form :
    \[
      G_h(x ; E)\,=\,\int e^{i (x\xi -F(\xi))} \sum_{k\geq 0} h^k b_k(\xi)\, d\xi. 
    \]
    A variant of the WKB method yields an eikonal equation for $F$, a homogenous transport equation for $b_0$ and inhomogenous transport equations for the $b_k, k\geq 1$. 

    Of course, the standard WKB construction near $b$ can be performed yielding a two-dimensional space of solutions. But the space of solutions that is $L^2$ near infinity is only $1$-dimensional. This difficulty is resolved by performing a stationary phase computation on the Maslov Ansatz. We obtain that there exists a unique solution, that is $L^2$ near infinity and such that 
    \begin{align*}
      & \begin{pmatrix}
        (E-V(b)^{\frac{1}{4}} G_h(b ; E)\\
        (E-V(b)^{-\frac{1}{4}} hG_h'(b ; E) 
      \end{pmatrix}
      =
      \cos\left( \frac{1}{h}\int_b^{+\infty} (E-V(y))_+^{\frac{1}{2}}\, dy -\frac{\pi}{4}\right )
        \cdot
        \begin{pmatrix}
          1+ \mathfrak{b}_h^+(E) \\
          \mathfrak{c}_h^+(E)\\
        \end{pmatrix}
        \\
       & \hspace{1cm} \,+\,\sin\left( \frac{1}{h}\int_b^{+\infty} (E-V(y))_+^{\frac{1}{2}}\, dy -\frac{\pi}{4}\right)
        \cdot
        \begin{pmatrix}
          \mathfrak{b}_h^-(E) \\
          1 + \mathfrak{c}_h^-(E)\\
        \end{pmatrix}
      \end{align*}
    where $\mathfrak{b}_h^{\pm}(E)$ and $\mathfrak{c}_h^{\pm}(E)$ are standard asymptotic expansions (with the positive integers as exponent set).   

    Combining with theorem \ref{thm:matching}, we obtain that there exist asymptotic expansions with exponent set
    $\{ m\gamma+n,~~m\geq 0, n\geq 0\}\setminus \{0\}$ such that following result holds:
    \begin{multline*}
      \begin{pmatrix}
        E^{\frac{1}{4}} G_h(0 ; E)\\
        E^{-\frac{1}{4}} hG_h'(0 ; E) 
      \end{pmatrix}
      \,=\,
      \cos\left( \frac{1}{h}\int_0^{+\infty} (E-V(y))_+^{\frac{1}{2}}\, dy -\frac{\pi}{4}\right )
        \cdot
        \begin{pmatrix}
          1+ \mathfrak{a}_h^+(E) \\
          \mathfrak{a}_h^+(E)\\
        \end{pmatrix}
        \\
        \,+\,\sin\left( \frac{1}{h}\int_0^{+\infty} (E-V(y))_+^{\frac{1}{2}}\, dy -\frac{\pi}{4}\right)
        \cdot
        \begin{pmatrix}
          \tilde{\mathfrak{a}}_h^-(E) \\
          1 + \tilde{\mathfrak{a}}_h^-(E)\\
        \end{pmatrix}.
      \end{multline*}

      \begin{rem}
        This formula allows one to recover the results of \cite{hillairet2023eigenvalue}
      \end{rem}

\bibliographystyle{alpha}
\bibliography{MMT-bib1}

\end{document}